\documentclass[reprint,amsmath,amssymb,notitlepage,superscriptaddress,nofootinbib,preprintnumbers,prl,twocolumn]{revtex4-2}
\usepackage{xcolor}
\usepackage{setspace}
\usepackage{ulem}
\usepackage{graphicx}
\usepackage{dcolumn}

\usepackage{bm}

\newtheorem{lemma}{Lemma}
\makeatletter
\newcommand{\xRightarrow}[2][]{\ext@arrow 0359\Rightarrowfill@{#1}{#2}}
\makeatother

\usepackage{hyperref}
\usepackage[english]{babel}
\usepackage{braket}
\usepackage{etoolbox}

\AtBeginDocument{%
  \providecommand{\language}[1]{}%
}

\begin{document}
\singlespacing
\title{Compressing Quantum Fisher Information}

\author{Rui Jie Tang}
\email{ruijie.tang@mail.utoronto.ca}
 \affiliation{CQIQC and Department of Physics, University of Toronto, 60 Saint George St., Toronto, ON M5S 1A7, Canada}

\author{Jeremy Guenza Marcus}
\email{jguenzama@gmail.com}
 \affiliation{CQIQC and Department of Physics, University of Toronto, 60 Saint George St., Toronto, ON M5S 1A7, Canada}
 
\author{Noah Lupu-Gladstein}
\affiliation{National Research Council of Canada, 100 Sussex Drive, Ottawa, Ontario K1N 5A2, Canada}

\author{Arthur O.T. Pang}
 \affiliation{CQIQC and Department of Physics, University of Toronto, 60 Saint George St., Toronto, ON M5S 1A7, Canada}
 
\author{C. Pria Dobney}
\affiliation{CQIQC and Department of Physics, University of Toronto, 60 Saint George St., Toronto, ON M5S 1A7, Canada}

\author{Giulio Chiribella}
\affiliation{QICI Quantum Information and Computation Initiative, School of Computing and Data Science, The University of Hong Kong, Pokfulam, Hong Kong}
\email{giulio@hku.hk}

\author{Aephraim M. Steinberg}
\affiliation{CQIQC and Department of Physics, University of Toronto, 60 Saint George St., Toronto, ON M5S 1A7, Canada}

\author{Y. Batuhan Yilmaz}
\affiliation{CQIQC and Department of Physics, University of Toronto, 60 Saint George St., Toronto, ON M5S 1A7, Canada}

\date{\today}

\begin{abstract}  We show that the quantum Fisher information about any phase parameter encoded in a family of pure quantum states can be faithfully compressed into a single qubit, accompanied by a logarithmic amount of classical bits.  When the phase is encoded into many identical copies of a  qubit state on the equator of the Bloch sphere, we show that the compression can be implemented sequentially, by iteratively compressing pairs of qubits into a single qubit. We experimentally demonstrate this building block  in a photonic setup, developing two alternative compression strategies, based on Type-I fusion gate and a postselected implementation of the {\tt CNOT} gate. 
\end{abstract}

\maketitle

\section{Introduction}

The quantum Fisher information (QFI) \cite{helstrom1969quantum,paris2009quantum} is a cornerstone of quantum sensing. It quantifies how sensitive a quantum state is to changes in an unknown parameter, with a high QFI indicating that a state undergoes a significant change even for a small variation in the parameter, thereby enabling more precise parameter estimation. Quantitatively, this relation is made precise  by the Quantum Cram\'er-Rao Bound (QCRB), which provides a fundamental limit  on measurement precision in terms of the QFI, and serves as a theoretical underpinning for the design of quantum sensors and quantum metrology setups \cite{giovannetti_advances_2011}. 



The typical process of quantum sensing involves preparing multiple copies of a probe in a quantum state that has high QFI for a given interaction with a system of interest,
 and progressively refining the information about the parameter through subsequent measurements. In applications such as remote sensing \cite{takeuchi2019quantum} and sensor networks \cite{komar2014quantum,proctor2018multiparameter}, however, the location where the sensors collect information and the location where the measurements extract this information may be different. It then becomes important to find ways  to efficiently transfer  the QFI from one location to another, and/or how to efficiently store at one location before the measurements take place.     
 
 A possible strategy to achieve efficient transfer/storage of the QFI is to compress the state of the sensors into a smaller number of (generally correlated) qubits. Specifically, Plesch and Bu\v zek \cite{plesch_efcient_2010} showed that  $N$ copies of an arbitrary pure state of two unknown parameters can be compressed into an exponentially smaller number of $\log_2(N+1)$  qubits, using the Schur-Weyl transformation \cite{bacon_efcient_2006}. An experimental proof-of-principle demonstration of this fact was provided by Rozema \textit{et al.} \cite{rozema_quantum_2014}.  
 Further research then extended the multi-copy compression  paradigm from pure to mixed states  \cite{yang2016efficient,yang2016optimal,yang_compression_2018}, eventually establishing the general compression limit of $\frac{f}{2}\log_2(N+1)$ \cite{yang_compression_2018-1} for $N$ copies  of a quantum state chosen from a family of states characterized by $f$ independent  parameters.   
  However,  preservation of the quantum state is not necessary for transferring/storing information about a specific physical parameter:  in principle, one could still achieve high sensitivity to that parameter without transferring/storing the full quantum state. For example, transferring information about the time parameter in the Hamiltonian evolution was shown to require a number of qubits that depends only on the accuracy achievable  with the best time measurement on the original state \cite{yang2018quantum}. The measure of accuracy adopted in this work, however, was different from the quantum Fisher information, and the question of the minimum number of qubits needed 
 to faithfully transfer/store the quantum Fisher information present in a quantum state remained unaddressed so far.   

 In this paper, we show that a \textit{single} qubit plus a logarithmic amount of classical information is sufficient to compress the quantum Fisher information about a phase parameter encoded in an arbitrary pure quantum state.  The compression protocol takes a particularly simple form in the case of multi-copy qubit states on the equator of the Bloch sphere, which have many applications in quantum metrology for precision measurements ranging from magnetometry via spin precession estimation \cite{baumgart2016ultrasensitive,danilin2018quantum} to precision interferometry \cite{degen2017quantum,patel2025review}.      We show that the compression protocol for multi-copy equatorial qubits can be decomposed into a tree of elementary compression protocols, in which the quantum Fisher information contained into two  qubits is compressed into the state of a single qubit.  In this way,  the original multi-copy state of $N$ qubits is compressed into a single qubit plus $\log_2(N-1)$ classical bits.   We experimentally demonstrate  the building block of this compression protocol on a photonic platform, introducing two alternative architecture to achieve the basic $2\to 1$ compression step using a Type-I fusion gate and a postselected {\tt CNOT}, respectively.

\section{Compressing the quantum Fisher information}
Consider a  family of quantum states of the form 
\begin{align}\label{states}|\Psi_\theta\rangle   =  e^{-i\theta H} |\Psi\rangle \,,
\end{align}
where $\theta \in \mathbb R$ is a real parameter, $|\Psi\rangle$ is a fixed pure state, and $H$ is a self-adjoint operator acting on the system's Hilbert space.  
The quantum Fisher information of the state $|\Psi_\theta\rangle$ is given by  \begin{align}
  F(|\Psi_\theta  \rangle) &= 4( \langle \partial_\theta\Psi_\theta  |\partial_\theta \Psi_\theta  \rangle-|\langle\Psi_\theta  |\partial_\theta\Psi_\theta  \rangle|^2) \, ,
     \label{qfi}
\end{align}
and provides a lower bound for the variance of any unbiased estimate of the parameter $\theta$ via the Quantum Cram\'er-Rao bound ${\rm Var}(\theta) \ge F  (|\Psi_\theta  \rangle)^{-1}$. 

In the following, we will show that the QFI of the states Eq. (\ref{states}) can be squeezed into a single qubit, accompanied by a logarighmic amount of classical information.  


\subsection{The case of multi-copy equatorial qubits}

To get started, we will provide an explicit QFI compression protocol in the special case where the  quantum system consists of $N$ qubits, with each qubit  prepared in the equatorial state $|e_\theta  \rangle=(|0\rangle+e^{i\theta}|1\rangle)/\sqrt 2$.  In this case,  the overall $N$-qubit state  has the tensor product form $|\Psi_\theta\rangle  = |e_\theta\rangle^{\otimes N}$, and Eq. (\ref{qfi}) yields  a QFI equal to $N$.

To construct the QFI compression protocol, we start from a basic building block that transforms two equatorial qubits into one, while preserving the QFI.  The bulding block consists into applying a {\tt CNOT} gate to the two qubits, and performing a computational basis measurement on the second qubit.  Applying  a {\tt CNOT} gate to  two equatorial qubits with  phases $\theta_1$ and $\theta_2$ yields the state 
\begin{align}
 {\tt CNOT}\,  (  |e_{\theta_1}\rangle \otimes |e_{\theta_2} \rangle)  &=\frac{ |e_{\theta_1 + \theta_2} \rangle   \otimes  |0\rangle  +  e^{i\theta_2}   |e_{\theta_1-\theta_2}\rangle  \otimes  |1\rangle}{\sqrt 2}\, .  \label{aftercnot}
\end{align}
Then, measuring the second qubit in the computational basis $\{|0\rangle, |1\rangle\}$ collapses the first qubit either in the state $|e_{\theta_1+\theta_2}\rangle$ or in the state  $|e_{\theta_2  -  \theta_1}\rangle$.  Overall, this basic protocol randomly performs either an addition or a subtraction of the phases, heralded by a measurement outcome.  \\

In the case $\theta_1 = \theta_2  = \theta$, the protocol either doubles the phase, or removes it altogether.  Notably, the QFI is preserved on average:  indeed,  the QFI of the state $|e_{2\theta}\rangle$ is $4\;\text{rad}^{-2}$ (as one can easily see from Eq. (\ref{qfi})), while the QFI of the state $|e_0\rangle$ is 0, and the average is $\frac 12 F  (|e_{2\theta}\rangle)  +  \frac 12 F  (|e_{0}\rangle) = 2\;\text{rad}^{-2}$, exactly equal to the QFI of the original two-qubit state $|e_\theta\rangle \otimes |e_\theta\rangle$. \\ 

To compress the QFI of $N$ qubits, we iterate the sum/difference protocol for $N-1$ times, using the first qubit as the control, as illustrated in Fig. \ref{fig:cnot_circuit}.  When the protocol is applied to $N$ equatorial qubits with phases $\theta_1, \theta_2 ,
\dots,  \theta_N$, it collapses the first qubit into the equatorial state with phase 
\begin{align}
\theta_{\rm tot} =\theta_1  +  \sum_{j=2}^N  \,  (-1)^{m_j}  \,  \theta_j \,,
\end{align}where $m_j \in  \{0,1\}$ is the outcome of the measurement on the $j$-th qubit.  In other words, repeated applications of the sum/difference protocol yield all possible linear combinations of the phases $\theta_1, \dots,  
\theta_N$ with $\pm 1$ coefficients.    

In particular, when the phases are all equal ($\theta_1  =  \theta_2  = \dots  = \theta_N$) and $k$ measurements yield outcome 0,  the protocol leaves the first qubit in the equatorial state with phase $(2 k  + 2  -  N) \theta$, which has QFI equal to $(2k+2-N)^2$.    Since the probability of obtaining outcome $0$ for $k$ times follows the binomial distribution 
\begin{align}
B_{N-1, k}   = \frac 1{2^{N-1}}  \, \begin{pmatrix}  N-1  \\  k  \end{pmatrix}\,,
\end{align}
the average QFI is $\sum_{k=0}^{N-1}  \,  B_{N-1,  k}   \,  (2k+2-N)^2=  N$, and coincides with the QFI of the original $N$-copy state. \\

In summary, all the QFI of the original $N$-copy state has been transferred to a single qubit, and can be accessed  by anyone who has access to the value of the measurement outcomes, or, more simply, to the value of the integer $k$.  Note that, since $k$ can take $N$ possible distinct values, our protocol compresses the QFI into a single qubit plus $\lceil \log_2 N \rceil$ bits of classical information.

\subsection{General Case}
 
We now extend the protocol from equatorial qubits to arbitrary quantum states of arbitrary (finite-dimensional) quantum systems.    The idea is to   decompose  the state $|\Psi_\theta\rangle  $ in Eq. (\ref{states}) as 
\begin{align}\label{psidecomp}
|\Psi_\theta\rangle    =  \sum_{E} \,  \sqrt{ p(E)}  \,  e^{-i\theta E}  |E\rangle    \, ,
\end{align}    
where $|E\rangle  $ is an eigenvector of the generator $H$ for eigenvalue $E$, and $p(E)$ is a probability distribution.  Using this decomposition, the QFI from Eq. (\ref{qfi}) can be written as
   \begin{align}
 F (  |\Psi_\theta\rangle   )      &  =  \sum_E   \,  p(E)  \,   E^2    -   \left(  \sum_E  \,  p(E) \,  E   \right)^2   \, .\label{qfipsi}
   \end{align}
  
Note that all the states $\{|\Psi_\theta\rangle  \}$ belong to  the subspace spanned by the vectors $\{|E\rangle  \}$.   We now construct a  measurement on this subspace, with measurement operators 
\begin{align}\label{Mk}
M_k  :=  \sum_E  \,  \sqrt{ p(  k|E) } \, |E\rangle  \langle  E| \qquad  k\in  \{1, \dots,  K\}\, ,
\end{align}  
where $k$ is the measurement outcome,  $K$ is the number of outcomes, 
and  $p(k|  E)$ is a conditional probability distribution, specified in the Appendix.  

When the measurement is performed on the state $|\Psi_\theta\rangle  $, the probability of the outcome $k$ is 
\begin{align}
p_k   : =  \|    M_k  |\Psi_\theta\rangle  \|^2    =  \sum_E    \,  p(k|E)  \,  p(E)  \,,
\end{align}
and, when $p_k$ is non-zero, the post-measurement state is   
\begin{align}
\nonumber |\Psi_{\theta, k}\rangle     & :  =   \frac{   M_{k}  |\Psi_\theta\rangle  }{\|  M_k |\Psi_\theta\rangle   \|  }   \\
&=    \sum_{E}     \sqrt{  p(E|k) }    \,  e^{-i\theta E}  \,        | E\rangle   \, ,
\end{align}
where we  defined $p (E|k)    : =   p(k|E) \,  p(E)/p_k$.  The QFI of this state is 
 \begin{align}
 F (  |\Psi_{\theta ,k}\rangle   )      &  =  \sum_E   \,  p(E|k)  \,   E^2    -   \left(  \sum_E  \,  p(E|k) \,  E   \right)^2   \, ,   \end{align}
 as one can see from Eq. (\ref{qfipsi}) by replacing $p(E)$ with $p(E|k)$.

On average, the QFI  of the post-measurement states  is
     \begin{align}
\nonumber F_{\rm av}  &=     \sum_{k=1}^K  \,  p_k\,   F  (|\Psi_{\theta, k}\rangle)  \\
&  =   \sum_E  p(E)   \,  E^2      -  \sum_{k=1}^K  p_k     \left(    \sum_E  p(E|k)   \,E    \right)^2 \, ,
   \end{align}
   Comparing this expression with Eq. (\ref{qfipsi}) we obtain that the average QFI  is equal to the QFI of the original state $|\Psi_\theta\rangle $ if and only if 
   \begin{align}\label{sameaverage}
 \sum_{k}  p_k      \left(    \sum_E  p(E|k)   \,E    \right)^2     =   \left( \sum_E  p(E)   \,E\right)^2            \, ,
     \end{align}
In the Appendix, we show that     it is  always possible to satisfy this condition with probability distributions $p(E|k)$ that have support on at most two distinct values of the energy.  In this way, every  state $|
   \Psi_{\theta,  k}\rangle  $  can be faithfully encoded   into a single qubit.  Moreover, we show that the number of measurement outcomes can be upper bounded as $K\le d-1$, where $d$ is the dimension of the subspace spanned by the eigenstates $\{|E\rangle\}$.   With the appropriate probability distributions $p(E|k)$ at hand, we can then reverse-engineer the probability distributions $p(k|E)$ in Eq. (\ref{Mk}), thereby determining the appropriate measurement operator $M_k$.\\

   All together, the above construction yields a measurement with $K\le d-1$ outcomes, 
   with the property that the conditional states can be faithfully encoded into a single qubit, and have on average the same QFI of the original state $|\Psi_\theta\rangle$.  Since the number of measurement outcomes is upper bounded by $d-1$, the amount of classical information needed to access the QFI is upper bounded by $\lceil \log_2 (d-1) \rceil$ bits. In summary, every one-parameter family of pure states of the form $|\Psi_\theta\rangle  = e^{-iH \theta} \, |\Psi\rangle$ can be compressed into a single qubit plus $\lceil \log_2 (d-1) \rceil$ bits in a way that preserves the total QFI.  
   
\begin{figure}[h]
    \centering
    \includegraphics[scale=0.75]{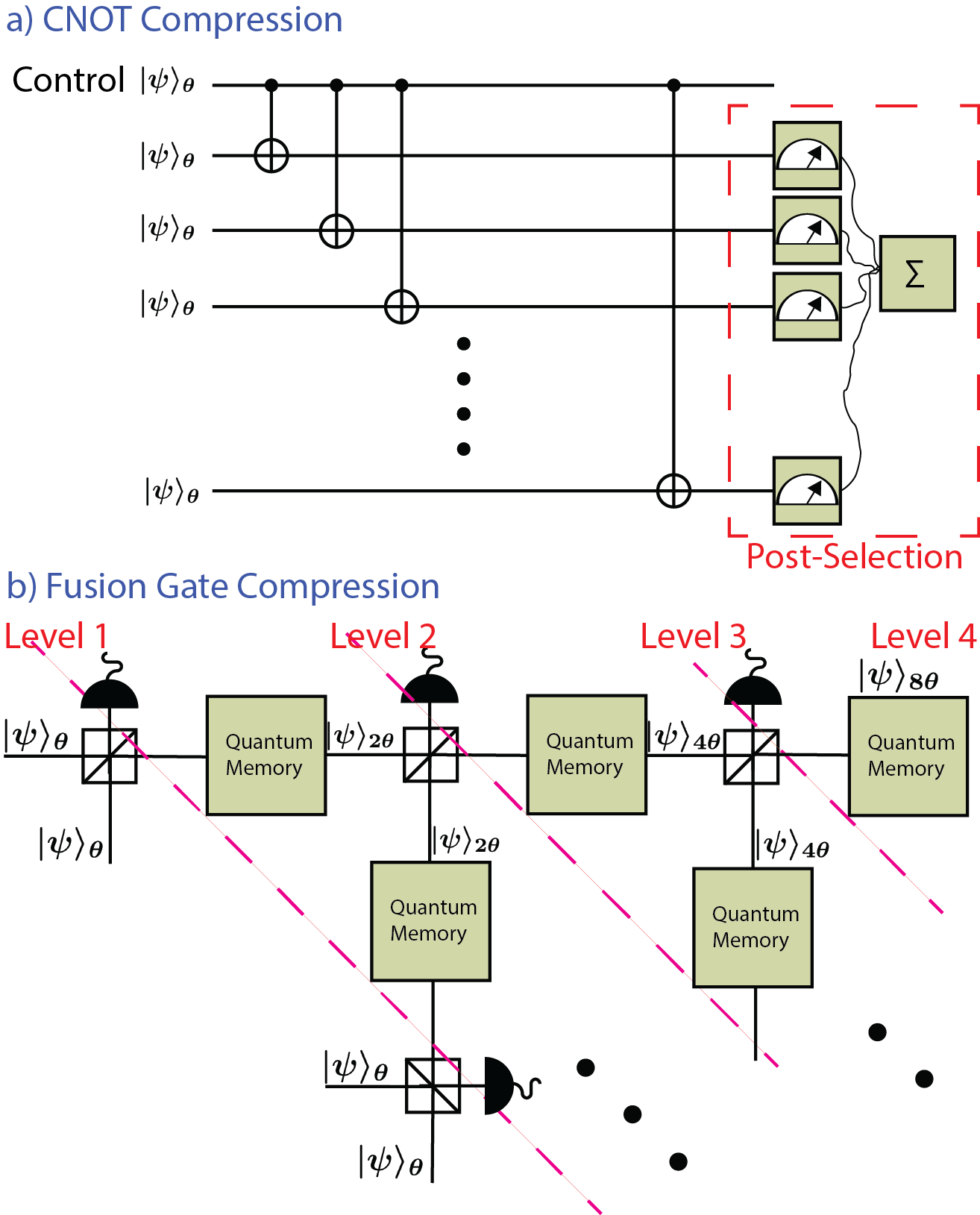}
    \caption{a) {\tt CNOT} cascade compression scheme that  compresses the QFI of arbitrary n-qubit inputs down to a single qubit. The qubit in the top path  acts as the control qubit for all subsequent {\tt CNOT} gates. Each {\tt CNOT} gate performs a two-qubit compression, and measuring the target qubits in the computational basis transfers all of the information to the control qubit.
    b) Type-I fusion gate compression scheme. Equatorial state $|\psi\rangle_{\theta}$ are input into Type-I fusion gates at the first level. If a single polarization qubit is detected at one of the PBS output ports, the qubit exiting the other port will contain the compressed QFI, with its phase doubled. Non-deterministically compressed qubits with the same phase from the previous level are sent into the next level for further compression. To ensure that two qubits are available at the same time, quantum memories need to be used as buffers to store the qubits until they are ready simultaneously. }
    \label{fig:cnot_circuit}
\end{figure}

\section{Experimental demonstration}
\subsection{CNOT Compression Scheme}

We now demonstrate the compression of the QFI for equatorial states using polarization qubits in a linear optical setup.  Our experimental scheme implements the basic building block for the compression of equatorial qubit states: the compression of two qubits into one using a {\tt CNOT} gate followed by a measurement on one of the qubits. 
 As shown in the previous section,  this basic building block  can then be cascaded to implement the compression of arbitrarily large numbers $N\ge2$ of equatorial qubits. \\

The setup is shown in Fig. \ref{fig:optical_setup}. The setup uses pairs of Type-I spontaneous parametric down-conversion (SPDC) 808 nm photons emitted at an opening angle of $3^\circ$ for both paths from a $1$ mm thick BBO nonlinear crystal. Our crystal is pumped by a $10\;\text{mW}$ 404 nm pulse. That is first prepared by using a mode-locked 140 fs Ti:Sapphire pulsed laser source emitting light at 808 nm and upconverting this pump to $404$ nm through photon upconversion with a BBO crystal. We detect $3310\pm58\;\text{pairs/sec}$ with Perkin-Elmer SPCM-AQ4C detectors. These detection signals are then fed into a home-built coincidence counter, which counts photon detection events across two detectors within a $4\;\text{ns}$ collection time window as coincidence events.

\begin{figure*}
\includegraphics[scale=0.7]{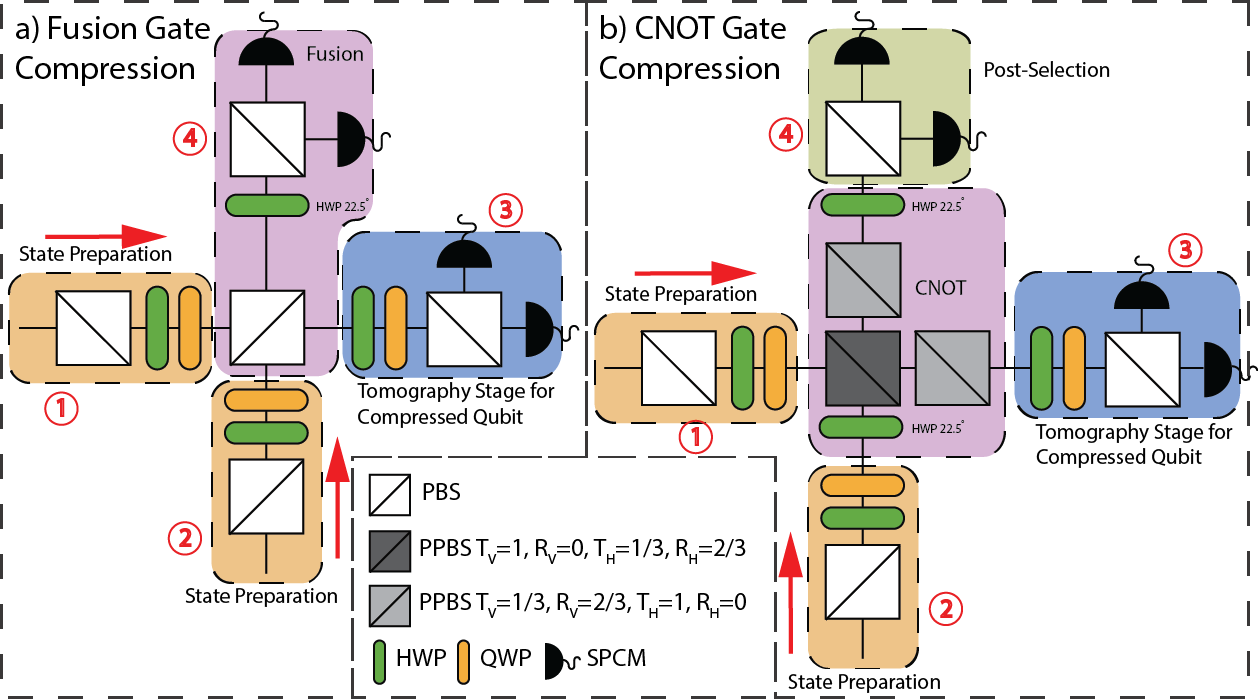}
    \caption{ Linear optical Implementation of QFI compression. Two photons generated via spontaneous parametric down-conversion (SPDC) serve as input polarization qubits for the compression schemes (a) and (b), entering through ports 1 and 2. The qubits are prepared in equatorial states with phase $\theta$ using a PBS, a HWP set to $\pi/8-\theta/4$, and a QWP set to $45^\circ$. A {\tt CNOT} gate is used to interfere the qubits, and post-selection occurs by projecting the qubit in path 4 on $|0\rangle_4$, after which the compressed qubit along path 3 is characterized with a tomography setup. }
    \label{fig:optical_setup}
\end{figure*}

Qubits in the experiment are the polarization degree of freedom of the pairs of down-converted photons. They prepared in identical equatorial states with phase $\theta$ using a combination of a polarizing beam splitter (PBS), a half-wave plate (HWP) and a quarter-wave plate (QWP) (Fig. \ref{fig:optical_setup}(a)).  \\

The {\tt CNOT} gate is realized  with partially polarizing beam splitters (PPBS), as proposed by Ralph et al. \cite{ralph_linear_2002}. This method is relatively simple to implement as it does not require auxiliary qubits, and successful {\tt CNOT} gate operations can be easily post-selected by observing two-photon coincidence events. We post-selected the successful compression results by projecting the target qubit onto the computation basis $|0\rangle\langle 0|_2$, which corresponds to $|H\rangle\langle H|_2$ in the physical basis. This projection is accomplished by detecting photons at the transmission port of the PBS in the target qubit path. Successfully detecting photons here projects the control qubit into its compressed state $|e_{2\theta}\rangle=\frac{1}{\sqrt{2}}(|H\rangle +e^{2i\theta}|V\rangle)$.

Next, we project this compressed qubit onto the diagonal basis $|\pm\rangle=\frac{1}{\sqrt{2}}( |H\rangle \pm |V\rangle)$. This projection is carried out using a HWP, a QWP, and a PBS. The measurement result, heralded by the target qubit, is expected to produces a sinusoidal function dependent on $\theta$ with a doubled frequency and halved period of $\pi$, which represents a clear signature of successful phase compression: 

\begin{align}
\text{Pr}_{\text{ideal}\pm}(\theta)=|\langle\pm|\psi\rangle_{2\theta}|^2 &= \frac{1}{2}(1 \pm \cos(2\theta))\label{eq:prpm-ideal} \, ,
\end{align}

However, experimental imperfections such as the photons not being perfectly indistinguishable, birefringent phaseshifts induced by PPBSs, and slight errors in the waveplate calibrations required us to modify this formula by adding extra parameters that depend on these imperfections.

\begin{align}
\text{Pr}_\pm(\theta)=|\langle\pm|\psi\rangle_{2\theta}|^2 &= \frac{1}{2}(1 \pm A\cos((2+\delta)\theta+\phi))\, ,
\end{align}

We prepared and compressed equatorial qubit states with varying phases $\theta$ ranging from $-90^\circ$ to $270^\circ$ in $2.5^\circ$ increments. For each phase setting, we collected photon-counting statistics for 30 seconds and repeated each measurement 10 times. Our results in Fig. \ref{fig:1d-plot}(a) show a clear double fringe pattern across a phase interval $2\pi$.

\begin{figure*}
\centering
\includegraphics[scale=0.4]{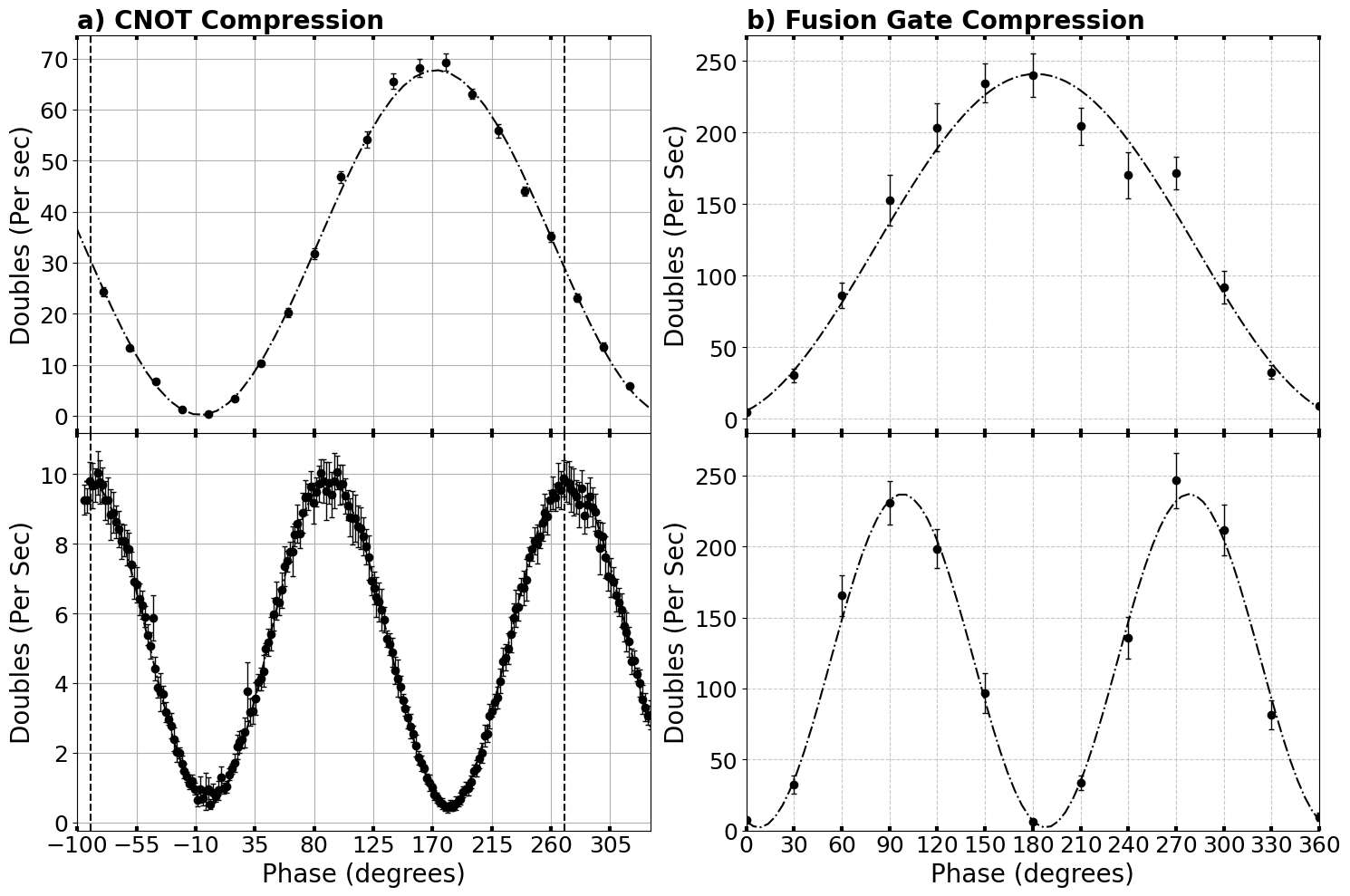}
\vspace*{-5mm}
\caption{ Output of successful  compression. Sample raw data for projecting compressed qubit $|\psi\rangle_{2\theta}=\frac{1}{\sqrt{2}}(|0\rangle+e^{2i\theta}|1\rangle)$ onto the diagonal basis $|+\rangle$.  {\tt CNOT} plots show detections over a one-second counting interval with phase increments of $2.5^\circ$ in $\theta$, set by rotating the HWPs in the state preparation setup Fig \ref{fig:optical_setup}. The single counts corresponds to uncompressed qubits $|\psi\rangle_\theta$ projected onto the diagonal basis with detector dark counts subtracted. The two-photon detection shows oscillations at a doubled frequency with error bars representing $\pm \sigma$ statistical uncertainty, mainly due to laser power fluctuation and Poissonian noise from non-deterministic photon pair generation. In general, fusion gate can add or subtract phases from two different equatorial qubits as shown in \cite{tham_experimental_2020}. A 2D scan demonstrating this is provided in Appendix. Here, this 1D plot shows the compression scenario where the two input phases are identical.}
\label{fig:1d-plot}
\end{figure*}

 We constructed an estimator that yields $\theta$ from the measurements:
\begin{align}
    \hat{\theta} &= \text{arccos}\left( \frac{\text{N}_+ - \text{N}_-}{A(\text{N}_+ + \text{N}_-)}\right) / (2+\delta)-\phi\, ,
\end{align}
Here, $N_\pm$ is the number of counts recorded for projectors $|\pm\rangle\langle\pm|$.
Each uncompressed qubit contains $1\;\text{rad}^{-2}$ of QFI.  While compressing 2 qubits, we expect the failed outcome to contain no information and the successful outcome to contain $4\;\text{rad}^{-2}$ of QFI. Therefore, when we estimate $\theta$ with $N$ successfully compressed qubits, the variance of the estimator, as predicted by the QCRB is the following:
\begin{align}
    \text{Var}(\theta) &= \frac{1}{4N}\;\text{rad}^{2}\\ 
    \sqrt{N} \text{Std}(\theta) &= \frac{1}{2}\;\text{rad}\, ,
\end{align}

 We demonstrated the compression from another perspective by estimating $\theta$ from the compressed qubits. Since the compressed qubits contain four times more information than uncompressed ones, the uncertainty of an estimator $\tilde\theta$ obtained from measuring $N$ compressed qubits should be half of that obtained from measuring $N$ uncompressed qubits. 

We compare the per photon root-mean-squared error (RMSE) and standard deviation against the Quantum Cram\'er-Rao Bound (QCRB) in Fig. \ref{fig:qfi}. The compressed qubits were measured in the optimal basis $\{\ket{\psi_{2\theta+\pi/2}},\ket{\psi_{2\theta-\pi/2}}\}$ to estimate $\theta$ with maximal precision. The measured standard deviations and RMSEs are multiplied by the square root of the number of detected photons, because the photon number varies throughout the experiment, and this normalization enable a fair comparison across all data points. The standard deviation of our measurements (square) closely follows the compressed QCRB limit of $\frac{1}{2}$, confirming that the QFI is successfully compressed and the phase sensitivity is enhanced by the expected factor of 4. 

We observed that the RMSE values (circles) are consistently higher than the standard deviation values (squares) in Fig. \ref{fig:qfi}. This discrepancy indicates the presence of a systematic bias in the estimator. Unlike statistical noise, which decreases with $1/\sqrt{N}$, systematic bias is constant with photon number and thus dominates the error in the high-$N$ regime. To investigate this, we decoupled the error contributions by calculating the systematic bias via $\text{Bias}=\sqrt{\text{RMSE}^2-\sigma^2}$ in Fig. \ref{fig:qfi}. We observed a phase-dependent bias oscillating with an amplitude of approximately $0.03$ rad. This bias is primarily attributed to a drift in photon visibility over time. The data were acquired sequentially, the phase estimations became coupled to this temporal drift, resulting in a biased Maximum Likelihood Estimation (MLE). To mitigate this effect in future work, data acquisition should be randomized in time.

\subsection{An Alternative Compression Scheme for Photonic Qubits: Type-I Fusion}

There exist multiple ways to implement {\tt CNOT} gates in the optical domain.  The downside of the realization of the {\tt CNOT} gate we implemented is that it only succeeds with a 1/9 probability. Due to this probabilistic nature, this implementation preserves the QFI only for the successful cases. In principle, higher success probability can be achieved with more complex linear optical setups using multiple auxiliary qubits, as shown  in   Refs.  \cite{knill_scheme_2001,pittman_probabilistic_2001,pittman_demonstration_2002,pittman_experimental_2003,okamoto_demonstration_2005,okamoto_realization_2011}. In particular, the protocol by Knill \textit{et al} \cite{knill_scheme_2001} enables a  fully deterministic realization of the {\tt CNOT} gate by using a quantum memory to store and teleport gate operations. 

Other platforms like atoms, superconductors etc. can implement deterministic {\tt CNOT} gates \cite{levine2019parallel,kandala2021demonstration} and would not be affected by this problem specific to optical platforms. However, it is possible to circumvent this problem in optical platforms by using an alternative compression scheme that preserves the average QFI with relatively simple operations. This alternative scheme uses linear optical Type-I fusion gates \cite{browne_resource-efficient_2005}, which have been used as a means for generating cluster states for measurement-based quantum computing \cite{bartolucci_fusion-based_2023} and adding the phases of equatorial qubits \cite{tham_experimental_2020}. 
 
 Our fusion-gate compression scheme is illustrated  in Fig. \ref{fig:optical_setup}(b).  The fusion gate starts with a polarizing beam splitter (PBS) which reflects vertical light and transmits horizontal light, thereby implementing the transformation 
\begin{align*}
|H\rangle_1  \otimes  |H\rangle_2 &\mapsto |H\rangle_3  \otimes  |H\rangle_4\\
|H\rangle_1  \otimes  |V\rangle_2 &\mapsto |H\rangle_3  \otimes  |V\rangle_3\\
|V\rangle_1  \otimes  |H\rangle_2 &\mapsto |H\rangle_4  \otimes  |V\rangle_4\\
|V\rangle_1  \otimes  |V\rangle_2 &\mapsto |V\rangle_3  \otimes  |V\rangle_4 \, ,
\end{align*}
where the subscripts $1$ and $2$ label the input ports, while the subscripts $3$ and $4$ label the output ports.  When  two equatorial polarization states with phases $\theta_1$ and $\theta_2$ enter the PBS  through input ports 1 and 2,  the output state at ports 3 and 4 is 
\begin{align}
 \nonumber   &{\tt PBS}  \,(  |e_{\theta_1} \rangle_1 \otimes |e_{\theta_2}\rangle_2)  \\
  \nonumber & =   \frac{1}{\sqrt 2} \,  \frac{ |H\rangle_3 \otimes |H\rangle_{4} +   e^{i  (\theta_1 + \theta_2)}\,  |V\rangle_3 \otimes |V\rangle_{4}}{\sqrt 2}  \\
     &    \quad +  \frac{e^{i\theta_2}}{\sqrt 2}    \, \frac{|H\rangle_{3} \otimes |V\rangle_3 + e^{i(\theta_1-\theta_2)}  \,|H\rangle_{4} \otimes |V \rangle_4 }{\sqrt 2}\, . \label{eq:after_pbs}
\end{align}
This output state  resembles the output state in Eq. (\ref{aftercnot}) in the sense that the state is decomposed into two terms: one carrying a relative phase of $\theta_1 +  \theta_2$, and the other carrying a relative phase of $\theta_1 - \theta_2$. Within our knowledge, there is no practical way to retrieve the information of both terms simultaneously. Hence, we post-select the branch where the photons are traveling along different paths, at the price of losing the photons when they are traveling on the same path: this procedure yields two photons carrying  phase $\theta_1+\theta_2$ with probability $1/2$,  and a discarded outcome with probability $1/2$.  Notably, when $\theta_1  =  \theta_2 =\theta$,  this procedure preserves the total QFI about the parameter $\theta$.   For the remainder of this section, we will assume  $\theta_1  =  \theta_2 =\theta$. 
 
 After the PBS, given that there is one photon on each path, the phase information carried by the two photons can be transferred to a single photon, by applying an HWP rotated to $22.5^\circ$ on one  photon (say, the photon on  path 4), and then measuring it on the $\{|H\rangle, |V\rangle\}$ basis.  The application of the rotated HWP produces the output state
  \begin{align}
\nonumber U_{\tt HWP}(22.5^0) & ~  \frac{|H\rangle_3 \otimes |H\rangle_{4}+e^{2i \theta }|V\rangle_3 \otimes |V\rangle_{4}}{\sqrt 2}\\
&\quad =  \frac{1}{\sqrt 2} \,  |e_{2\theta} \rangle_3 \otimes  |H\rangle_4+  \frac 1{ \sqrt 2} \, |e_{2\theta+\pi}\rangle_3\otimes |V\rangle_4 \, ,
\end{align}
and the measurement on the second photon steers the first photon to either the equatorial  state $|e_{2\theta}\rangle$  or to the equatorial state $|e_{2\theta+\pi}\rangle$.  

\begin{figure*}
\includegraphics[scale=0.7]{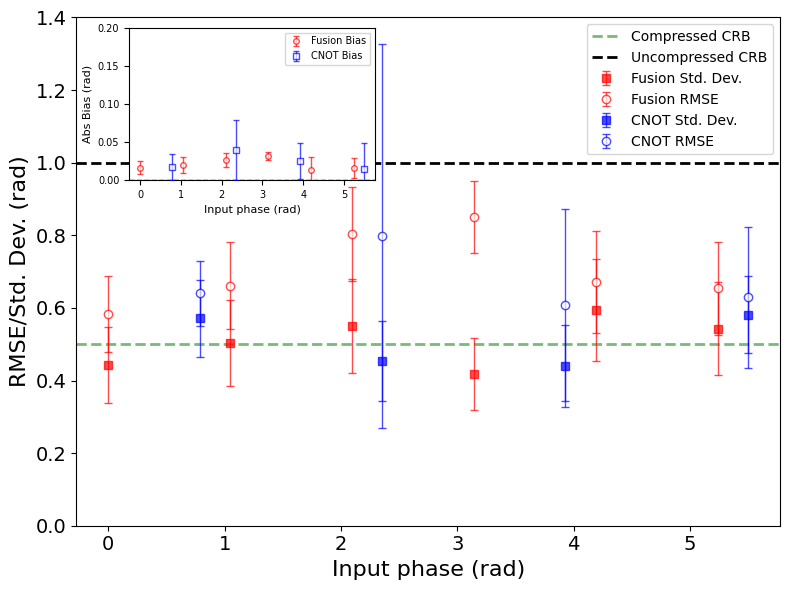}
\caption{ The standard deviation and the root-mean-squared error (RMSE) of the estimated phase. The green dashed line represents the Cram\'er-Rao bound for estimating $\theta$ using and uncompressed qubit $|\psi\rangle_\theta$, while the black dashed line corresponds to the standard deviation for a compressed qubit $|\psi\rangle_{2\theta}$. The squares indicates the experimentally observed variances when estimating a selected set of phases, and the circular dots indicate the per photon RMSEs. The latter are higher due to the bias being scaled up by $\sqrt{N}\approx16.64$ for a mean photon count of $N=277$ for the CNOT data and $N \approx 522$ for Fusion data (see inset). Red markers show the experimental data for phase estimations using compressed qubits from the fusion gate setup, and blue markers represent data from the {\tt CNOT} compression setup (Detailed descriptions of how we constructed the estimators are provided in the main text). The statistical standard deviation was estimated by performing Monte Carlo Simulations on the phase estimator using measured mean photon numbers with Poissonian noise. Inset: The bias, calculated as $\text{Bias}=\sqrt{\text{RMSE}^2-\sigma^2}$, reveals a systematic error of $0.02 - 0.03$ rad ($1^\circ-2^\circ$), which is attributed to imperfect waveplate retardances and drifting visibility over the course of data collection.}
\label{fig:qfi}
\end{figure*}

The above protocol for compressing the QFI from two polarization qubits to one. 
This protocol can be used as the basic  building block for a protocol that compresses the QFI of $N$ qubits into a variable number of qubits, between 1 and  $\lfloor\log_2 N\rfloor$.     The protocol works iteratively as follows: first, divide the initial $N$ qubits into $\lfloor N/2 \rfloor$ pairs,  with one photon remaining unpaired if $N$ is odd.   Then, perform the basic fusion-gate protocol for each pair. This procedure yields a $N'$ qubits with phase $2\theta$, where $N'$ is an integer between $0$ and $\lfloor N/2 \rfloor$. If $N'\ge 2$, then the procedure is iterated on the $N'$ qubits.  It is easy to see that the iteration concludes in at most $\lfloor \log_2 N \rfloor$ steps. Since in each step at most one qubit is left unprocessed, the number of qubits remaining in the end is at most  $\lfloor \log_2 N \rfloor$.  In the best scenario, when $N$ is a power of 2 and all the fusion steps are successful, the protocol outputs a single qubit. In the worst case the compression rate is the same as the rate for the compression of $N$-qubit states in the symmetric subspace. Additionally, for each remaining qubit, a classical bit needs to be stored to keep track of extra $\pi$ phaseshifts.

The experimental implementation of the fusion-gate compression scheme uses the same state preparation procedures as  the {\tt CNOT} implementation.   We demonstrate the building block of $2\to 1$ compression, using a single PBS for the fusion gate.  
Successful compression is identified by detecting coincidence photon events where the two photons took different paths within the setup and projecting the compressed qubit onto the diagonal basis. This allowed us to observe a two-fringe projection pattern, shown in Fig. \ref{fig:1d-plot}(b), similar to what we have observed in the {\tt CNOT} compression data.\\ 

\section*{Conclusion}

In conclusion, we have presented two protocols for compressing quantum Fisher information (QFI) from multiple $N$ single-parameter equatorial qubits into a single qubit and $\log_2(N-1)$ classical bits. We experimentally demonstrated these protocols with a linear optical setup for compressing information of the phase from two qubits to a single one. We discussed how our two-qubit compression technique can be cascaded to compress an arbitrary number of qubits. We verified the compression of QFI by estimating the parameter from the compressed qubit and comparing the result to the Cram\'er-Rao bound. This technique can be used to transfer quantum information using less resources and can be useful in scenarios such as remote and distributed sensing where identical states going through the same interaction can be compressed and transmitted using less resources. 

\bibliography{references2}

\appendix 

\section{Choice of the conditional probability distribution in Eq. (\ref{Mk}).}
Here we show that the condition 
 \begin{align}\label{sameaverage1}
 \sum_{k=1}^K  p_k      \left(    \sum_E  p(E|k)   \,E    \right)^2     =   \left( \sum_E  p(E)   \,E\right)^2            
     \end{align}
     can be satisfied by  probability distributions $\{p(E|k)\}_{k=1}^K$ that have  support on at most two distinct values of $E$.  This result is proven by suitably choosing the conditional probability distribution $p(k|E)$ in Eq. (\ref{Mk}).

Consider the  set $\sf S$ of all functions  $q(E)$  satisfying the conditions  
\begin{align}\label{convex1}
 &q  (E) \ge 0  \qquad \forall E  \\
\label{convex2} &\sum_E q(E)  =1 \\
\label{convex3}&\sum_E   q(E)E   =  \epsilon  \, , 
\end{align}
with  $\epsilon  :=  \sum_E   p(E) \,  E$.   In other words, the set $\sf  S$ consists of probability distributions $q(E)$ that have the same average as the probability distribution $p(E)$ (defined in Eq. (\ref{psidecomp})). Notice that, obviously, $p(E)$ is itself an element of the set $\sf S$.

Eqs. (\ref{convex1})-(\ref{convex3}) show that $\sf  S$ is a compact convex set of dimension $d-2$, where $d$ is the dimension of the vector space spanned by the vectors $\{|e\rangle\}$.  By Caratheodory's theorem, every element of  a compact $(d-2)$-dimensional convex set can be decomposed into a mixture of at most $d-1$ extreme points. 
In particular,  the probability distribution $p(E)$ can be decomposed as 
\begin{align}\label{convexdecomp}
p  (E)  =  \sum_{k=1}^{K}   \,    p_k\,  p(E|k)  \, ,
\end{align}
where $K$ is upper bounded by $d-1$, and, for every $k\in  \{1,\dots, K\}$,    $p(E|k)$ is an extreme element of the set $\sf S$.  

We now show that the extreme elements of $\sf S$ are probability distributions $q(E)$ that assign non-zero probability to at most two values of $E$. 

\begin{lemma}
Let $q(E)$ be an extreme element of the convex set defined by Eqs. (\ref{convex1})-(\ref{convex3}). Then, there exist at most two distinct values $E_0$ and $E_1$ such that $q(E_0)\not =  0$ and $q(E_1)\not =  0$.
\end{lemma}

{\bf Proof.} If there exists a non-zero function $f(E)$ and a positive number $\lambda>0$ such that such that $  q(E)  \pm \lambda\,   f(E)$ are both elements of $\sf  S$, then clearly $q(E)$ is not an extreme element, because one would have $q(E)  =  \frac 12  \,  (q(E) +  \lambda\, f(E))+  \frac 12  \,  (q(E) -  \lambda\, f(E))$.    This sufficient condition for non-extremality is satisfied if and only if
\begin{align}\label{positivity}
  &-  q(E)  \le  \lambda\,   f(E)    \le  q(E)    \qquad\forall E \\
 &\sum_E f(E)  = 0   \label{convex20} \\
&\sum_E   E\,  f(E)  = 0  \label{convex30}  \, , 
\end{align}

Eq. (\ref{positivity}) is equivalent to the condition 
\begin{align}\label{support}
f(E) \not  = 0    \qquad {\rm only~ if}  \qquad  q(E)  \not  = 0  \, . 
\end{align}
Indeed, if this condition is satisfied, one can always find a $\lambda>  0$ such that  Eq. (\ref{positivity}) is also satisfied.   The functions satisfying condition (\ref{support}) form a real vector space, whose dimension is equal to the number of values of $E$ such that $q(E)  \not  = 0$.  Let us denote this vector space by ${\cal V}_q$ and write the function $f  (E)$ as a vector $\bf f$ with entries labeled by $E$.

Conditions (\ref{convex20}) and (\ref{convex30}) can be written as $  {\bf f } \cdot {\boldsymbol 1}  =   0$ and  ${\bf  f} \cdot  {\bf  E}  =  0$, respectively,  where $\boldsymbol 1$ is the  vector with all entries equal to 1,  and $\bf E$ is the vector with the $E$-th entry equal to  $E$.   
In other words,  $\bf f$ must be orthogonal to the two vectors $\boldsymbol 1$ and $\bf E$.  

Note that, if the vector space ${\cal V}_q$ has dimension larger than 2, then will be surely a vector  $\bf f$  that is orthogonal to $\boldsymbol 1$ and $\bf E$. In this case, $q(E)$ is non-extreme.

Summarizing, if ${\cal V}_q$ has dimension larger than 2, then $q (E)$ is non-extreme.  
Conversely, if $q(E)$ is extreme, then ${\cal V}_q$ must have dimension smaller than or equal to 2.  Recall that the dimension of ${\cal V}_q$ is the number of values of $E$ such that $q(E)\not= 0$.  Hence, we obtained that an extreme probability distribution $q(E)$ can be non-zero on at most 2 values.    $\blacksquare$  

\medskip 
  
 We now apply the above lemma to the probability distributions $p(E|k)$ in Eq. (\ref{convexdecomp}).  Let us denote by $E_{0,k}$ and  $E_{1,k}$ the two values such that    $p(E_{0,k}|  k) >0$ and  $p(E_{1,k}|  k) >0$.   Since the state $|\Psi_{\theta,k} \rangle $  is given by  
 \begin{align}
 |\Psi_{\theta, k}\rangle    =  \sum_E  \,\sqrt{ p(E|k) }\,  e^{i  \theta E }\,  |E\rangle \, ,   \end{align}
 we obtain 
  \begin{align}
\nonumber  |\Psi_{\theta,k}\rangle    & =  \sqrt{  p( E_{0,k} | k )}  \,  e^{-i \theta  E_{0,k}}  \,  |E_{0,k}\rangle \\
& \qquad +  \sqrt{  p( E_{1,k} | k )}  \,  e^{-i \theta  E_{1,k}}  \,  |E_{1,k}\rangle   \, . 
 \end{align}
This state can be transformed into the qubit state
\begin{align}
 |\psi_{\theta,k}\rangle     =  \sqrt{  p( E_{0,k} | k )}  \,  e^{-i \theta  E_{0,k}}  \,  |0\rangle    +  \sqrt{  p( E_{1,k} | k )}  \,  e^{-i \theta  E_{1,k}}  \,  |1\rangle    
 \end{align}
 by applying the operation   $W_k :  =    |0\rangle   \langle   E_{0,k}|   +  |1\rangle  \langle    E_{1,k}|$.     This operation is simply an encoding of the two-dimensional subspace spanned by the vectors $|E_{0,k}\rangle$ and $|E_{1,k}\rangle$ into the state space of a single qubit.  As such, it preserves the quantum Fisher information for all the states in the subspace.  

 Finally, we observe that the probabilities $p_k$ and the conditional probability distributions $p(E|k)$ appearing in the convex decomposition (\ref{convexdecomp}) determine the conditional probability distributions $p(k|E)$ needed to define the measurement operators $M_k$ in Eq. (\ref{Mk}).   The probability distributions are defined in the obvious way, as $p(k|E)  :  =  p(E|k) p_k/p(E)$, where $p(E)$ is the probability distribution appearing in Eq. (\ref{psidecomp}).

\end{document}